\documentclass[journal]{IEEEtran}

\usepackage{amsmath,amssymb,amstext,mathtools,amsfonts,mathpazo,mathtools,amsthm}
\mathtoolsset{showonlyrefs}

\usepackage{algorithmic}
\usepackage{graphicx}
\usepackage{textcomp}

\usepackage[justification=centering]{caption}
\usepackage{eqparbox}
\usepackage{ulem}
\usepackage{cancel}
\usepackage{pdfpages}
\usepackage{epstopdf}
\usepackage{placeins}
\usepackage{aligned-overset}

\usepackage{times}
\usepackage{graphicx}
\usepackage{amssymb}
\usepackage{url,hyperref}

\usepackage[
backend=biber,
style=ieee,
sorting=none,
date=year,
doi=false,
url=false,
isbn=false
]{biblatex}

\newcommand{\sign}[1]{\text{sign}\left( #1 \right)}
\newcommand{\sat}[1]{\text{sat}\left( #1 \right)}
\newcommand{\etal}{\textit{et al.}}
\newcommand{\wrt}{w.r.t.}
\newcommand{\Rkreis}{\textsuperscript{\textregistered}}
\newcommand{\sgnpow}[2]{\lfloor #1 \rceil^#2}

\newcommand{\Brogliato}{\text{Brogliato}}
\newcommand{\Koch}{\text{Koch}}
\newcommand{\Xiong}{\text{Xiong}}
\newcommand{\Hanan}{\text{Hanan}}
\newcommand{\proposed}{\text{proposed}}
\newcommand{\explicit}{\text{explicit}}

\newcommand{\nonum}{\nonumber \\}

\newtheorem{theorem}{\textbf{Theorem}}
\newtheorem{proposition}[theorem]{\textbf{Proposition}}

\newcounter{casecounter}
\newcounter{subcasecounter}
\renewcommand{\thecasecounter}{\arabic{casecounter}}
\renewcommand{\thesubcasecounter}{\arabic{casecounter}.\alph{subcasecounter}}
\newenvironment{textcases}[0]{}{
	\setcounter{casecounter}{0}
}
\newenvironment{case}[2]{
	\refstepcounter{casecounter}
	\ifnum\pdfstrcmp{#2}{}=1
	\label{#2}
	\fi
	\noindent\textbf{Case \thecasecounter:} #1
}{
	\setcounter{subcasecounter}{0}
}
\newenvironment{subcase}[2]{
	\refstepcounter{subcasecounter}
	\ifnum\pdfstrcmp{#2}{}=1
	\label{#2}
	\fi
	\textbf{Case \thesubcasecounter:} #1
}{}

\makeatletter
\def\underbracex#1#2{\mathop{\vtop{\m@th\ialign{##\crcr
				$\hfil\displaystyle{#2}\hfil$\crcr
				\noalign{\kern3\p@\nointerlineskip}%
				#1\crcr\noalign{\kern3\p@}}}}\limits}

\def\upbracefilla{$\m@th \setbox\z@\hbox{$\braceld$}%
	\bracelu\leaders\vrule \@height\ht\z@ \@depth\z@\hfill 
	\kern\p@\vrule \@width\p@\kern\p@\vrule \@width\p@\kern\p@\vrule \@width\p@
	$}

\def\upbracefillb{$\m@th \setbox\z@\hbox{$\braceld$}%
	\vrule \@width\p@\kern\p@\vrule \@width\p@\kern\p@\vrule \@width\p@\kern\p@
	\leaders\vrule \@height\ht\z@ \@depth\z@\hfill\bracerd
	\braceld\leaders\vrule \@height\ht\z@ \@depth\z@\hfill
	\kern\p@\vrule \@width\p@\kern\p@\vrule \@width\p@\kern\p@\vrule \@width\p@
	$}

\def\upbracefillc{$\m@th \setbox\z@\hbox{$\braceld$}%
	\vrule \@width\p@\kern\p@\vrule \@width\p@\kern\p@\vrule \@width\p@\kern\p@
	\leaders\vrule \@height\ht\z@ \@depth\z@\hfill
	\kern\p@\vrule \@width\p@\kern\p@\vrule \@width\p@\kern\p@\vrule \@width\p@
	$}

\def\upbracefilld{$\m@th \setbox\z@\hbox{$\braceld$}%
	\vrule \@width\p@\kern\p@\vrule \@width\p@\kern\p@\vrule \@width\p@\kern\p@
	\leaders\vrule \@height\ht\z@ \@depth\z@\hfill\braceru$}

\def\upbracefillbd{$\m@th \setbox\z@\hbox{$\braceld$}%
	\vrule \@width\p@\kern\p@\vrule \@width\p@\kern\p@\vrule \@width\p@\kern\p@
	\bracerd\braceld
	\leaders\vrule \@height\ht\z@ \@depth\z@\hfill\braceru$}
\makeatother

\newcommand{\revision}[1]{#1}

\begin{document}
	\title{Modified Implicit Discretization of the Super-Twisting Controller}
	\author{Benedikt Andritsch, Lars Watermann, Stefan Koch, Markus Reichhartinger, Johann Reger, \\and Martin Horn
	\thanks{
		This work has been submitted to the IEEE for possible publication. Copyright may be transferred without notice, after which this version may no longer be accessible.\\
		The financial support by the Austrian Science Fund (FWF) grant no.
		I 4152, the Austrian Federal Ministry for Digital and Economic Affairs,
		the National Foundation for Research, Technology and Development and
		the Christian Doppler Research Association is gratefully acknowledged.
		The second and fifth author further acknowledge the financial support by
		the German Research Foundation (DFG), project no. 416911519, and by
		the European Union Horizon 2020 research and innovation program under
		Marie Skłodowska-Curie grant no. 824046.}
	\thanks{B. Andritsch and M. Reichhartinger are with the Institute of Automation
		and Control at Graz University of Technology, Graz, Austria.
		(email: benedikt.andritsch@tugraz.at, markus.reichhartinger@tugraz.at)}
	\thanks{S. Koch and M. Horn are with the Christian Doppler Laboratory
		for Model-Based Control of Complex Test Bed Systems, Institute of
		Automation and Control, Graz University of Technology, Graz, Austria.
		(email: stefan.koch@tugraz.at, martin.horn@tugraz.at)}
	\thanks{L. Watermann and Johann Reger are with the Control Engineering Group at
		Technische Universität Ilmenau, Ilmenau, Germany. (email: lars.watermann@tu-ilmenau.de, johann.reger@tu-ilmenau.de)}
	}
	
	\maketitle
	
	\graphicspath{{figures}} 
	
	\begin{abstract}
		In this paper a novel discrete-time realization of the super-twisting controller is proposed. The closed-loop system is proven to \revision{converge to an invariant set around the origin in finite time.}
		Furthermore, the steady-state error is shown to be independent of the controller gains. It only depends on the sampling time and the unknown disturbance.
		The proposed discrete-time controller is \revision{ evaluated comparative} to previously published discrete-time super-twisting controllers by means of the controller structure \revision{and in} extensive simulation studies.
		The continuous-time super-twisting controller is capable of rejecting any unknown Lipschitz-continuous perturbation \revision{and converges in finite time}. Furthermore, the convergence time decreases, if any of the gains is increased.
		\revision{The simulations demonstrate that the closed-loop systems with each of the known controllers lose one of these properties, introduce discretization-chattering, or do not yield the same accuracy level as with the proposed controller.}
		The proposed controller, in contrast, is beneficial in terms of the above described properties.
	\end{abstract}
	
	\section{Introduction}
	\label{sec:introduction}
	The field of Sliding Mode (SM) Control (SMC) has proven to be of high importance when considering systems with unknown disturbances~\cite{Shtessel2013}. In continuous-time, SMC manages to completely reject any disturbances that fulfill some requirements like boundedness or Lipschitz-continuity. However, SM controllers are mostly implemented on discrete-time hardware, requiring appropriate representations of these controllers. Discrete-time SM controllers have to deal with unpleasant effects like discretization-chattering, which diminishes the advantageous properties of SMC~\cite{Acary2012, Levant2011}. One of the first techniques in conventional SMC~\cite{Shtessel2013} avoiding discretization-chattering is the implicit discretization~\cite{Acary2010}. 
	
	A famous continuous-time SM system is the Super-Twisting Algorithm (STA)~\cite{LEVANT1993, Moreno2012}. 
	The STA is capable of rejecting Lipschitz-continuous disturbances, which is of high interest in real-world control problems~\cite{Chen2016, Pizzo2017}. Therefore, a proper discrete-time implementation of the Super-Twisting Controller (STC) is essential for many applications.
	There have been different approaches to achieve an implicitly discretized version of the STC~\cite{Brogliato2018,Xiong2022}. Also, non-implicit discretization techniques \revision{can be} applied to the STC, e.g. the matching approach~\cite{Koch2019} \revision{and the low-chattering discretization~\cite{Hanan2022}}.
	
	\revision{The continuous-time STC has the following properties:
	\begin{itemize}
		\item reject Lipschitz-continuous perturbations~\cite{LEVANT1993},
		\item finite-time convergence~\cite{LEVANT1993} and
		\item increasing any controller parameter reduces the convergence time~\cite{Utkin_2013}.
	\end{itemize}
	Furthermore, the following are desired properties of the discrete-time controller:
	\begin{itemize}
		\item no discretization-chattering occurs, i.e. the control error vanishes when no disturbance and no measurement noise are present,
		\item the steady-state error, i.e. the control error after all transients died out, is proportional to the discretization time squared, as in~\cite{Livne_2014}, and
		\item the steady-state accuracy is insensitive to controller parameters.
	\end{itemize}
	With the last property, the controller parameters can be selected solely on requirements regarding convergence time and control signal magnitude, and do not have to consider a trade-off with the accuracy of the controller.}
	\revision{Each of the existing discretizations of the STC fails to resemble some of these properties. Therefore, in this paper a novel discretization of the STC is presented, that unites all above-mentioned features. Measurement noise is not investigated in this paper.}
	
	\revision{At first}, an overview of existing discrete-time versions of the STC is given in Section~\ref{sec:related-work}. Then, in Section~\ref{sec:proposed} a novel implicit discretization of the STC is presented. Further, stability properties of the presented discretization are analyzed. Finally, in extensive simulation studies in Section~\ref{sec:simulation} it is demonstrated that the proposed controller preserves all crucial properties of the continuous-time STC, in contrast to previously published controllers.
	
	\subsection*{Mathematical Notations}
	Let 
	\begin{align}
		\sign{x} \in \begin{cases}
			\revision{\left\{1\right\}} & \revision{\text{if } x > 0,} \\
			\revision{\left\{-1\right\}} & \revision{\text{if } x < 0,} \\
			[-1,1] & \text{if } x = 0
		\end{cases}
	\end{align}
	be the signum function with $x \in \mathbb{R}$. Furthermore, the signed power function
	$
		\sgnpow{x}{y} = \sign{x} |x|^y,
	$
	with $x, y \in \mathbb{R}$ will be used. Note that $\sgnpow{x}{0} = \sign{x}$.
	Finally, let
	$
		\sat{x} = \begin{cases} x \quad \text{if } |x| < 1 \\ \sign{x} \quad \text{else} \end{cases}
	$ be the saturation function.

	\section{Related Work}\label{sec:related-work}
	
	The STC considers the \revision{dynamics of the sliding variable $x_1$ with relative degree~1 of an affine-input dynamic system~\cite{LEVANT1993},~\cite{SlidingModes2012}, i.e.}
	\begin{align}\label{ct-plant}
		\dot x_1 &=u + \varphi, \nonum
		\dot \varphi &=\Delta.
	\end{align}
	\revision{System~\eqref{ct-plant} is denoted as plant and $x_1$ and the perturbation $\varphi$ as plant states in the following. The remaining terms are} the control signal $u$ and the \revision{Lebesgue-measurable} unknown disturbance $\Delta(t)$ with $|\Delta\revision{(t)}| < L~\revision{\forall t}$ \revision{and} some \revision{known} constant $L$. \revision{Due to the bounded derivative} $\varphi$ is Lipschitz-continuous.
	
	The dynamic \revision{sliding-mode} controller 
	\begin{align}\label{ct-controller}
		u &=-\alpha \sgnpow{x_1}{\frac{1}{2}} + \nu, \nonum
		\dot \nu &=-\beta \sgnpow{x_1}{0},
	\end{align}
	with the controller state $\nu$ is known as STC and stabilizes $x_1 = 0$ of~\eqref{ct-plant} if the constant gains $\alpha$ and $\beta~\revision{ > L}$ are chosen accordingly~\cite{LEVANT1993}.
	The closed-loop system 
	\begin{align}\label{ct-sta}
		\dot x_1 &=-\alpha \sgnpow{x_1}{\frac{1}{2}} + x_2, \nonum
		\dot x_2 &=-\beta \sgnpow{x_1}{0} + \Delta,
	\end{align}
	with \revision{$x_2 = \varphi + \nu$} resulting from the plant~\eqref{ct-plant} and the STC~\eqref{ct-controller} is called STA.
	
	The goal of implementing the STC is to determine a discrete-time representation of the controller~\eqref{ct-controller}, which can generally be written as
	\begin{align}\label{dt-general-controller}
		u_k &=-\alpha \Psi_1(x_{1,k}) + \nu_{k+1}, \nonum
		\nu_{k+1} &=\nu_k - h \beta \Psi_2(x_{1,k}),
	\end{align}
	with the state dependent functions $\Psi_1$ and $\Psi_2$, the \revision{constant} discretization-time $h$, the known discrete-time system state $x_{1,k} = x_1 (k h)$, and $k=1,2,\dots$. \revision{Note that measurement noise in the state $x_{1,k}$ is not investigated in this paper.} The discrete-time control variable $u_k$ is then fed to the continuous-time system through a zero-order hold element, i.e. $u(t) = u_k \text{ for } kh \leq t < (k+1)h$. Note that $u_k$ contains the controller-state at $k+1$, i.e. $\nu_{k+1}$, as in this paper mainly implicit discretization approaches are considered. 
	Denote by $\mathcal{C}_*$ the controller resulting from~\eqref{dt-general-controller} and specific controller functions $\Psi_{1,*}$ and $\Psi_{2,*}$.
	In the following, several discrete-time realizations of the STC are presented.
	
	\subsection{Implicit Discretization}
	One discrete-time STC was published in~\cite{Brogliato2018, Brogliato2020} by Brogliato \etal\ and can be written as
	\begin{align}\label{dt-brogliato}
		&\revision{\Psi_{1,\Brogliato}} =\nonum &=\revision{\sign{x_{1,k}} \left( -\frac{h \alpha}{2} + \sqrt{\frac{h^2 \alpha^2}{4} + \text{max}(0, |x_{1,k} + h \nu_k| - h^2 \beta)}\right)}, \nonum
		&\Psi_{2,\Brogliato} = \sat{\frac{x_{1,k} + h \nu_k}{h^2 \beta}}.
	\end{align}
	Note that $\Psi_{1,\Brogliato}$ depends not only on the plant state $x_{1,k}$, but also on the controller state $\nu_k$. The authors use an implicit discretization approach to establish the explicitly given discrete-time controller functions~\eqref{dt-brogliato}. It is proven that the undisturbed closed-loop system is globally asymptotically stable. The authors introduce the sliding variable $x_{1,k} + h \nu_k$, which is driven to zero and maintained there. Also, $\mathcal{C}_{\Brogliato}$ drives the closed-loop state $\varphi_k + \nu_k$ to zero. 
	
	However, let us assume an unbounded perturbation $\varphi_k$, e.g. due to a constant disturbance $\Delta_k$. Then $\varphi_k$ will grow, and thus, also $\nu_k$ will grow. With the sliding variable kept at the origin, therefore also $x_{1,k}$ will grow and the control goal $x_{1,k} = 0$ can not be maintained. Therefore, the controller $\mathcal{C}_\Brogliato$ is not able to reject an unbounded perturbation $\varphi$, which reduces the class of disturbances $\Delta$ that can be handled by the controller compared to the continuous-time STC. These thoughts will be discussed in simulations in Section~\ref{sec:simulation} as well.
	
	\subsection{Discretization Based on Matching Approach}
	Another discrete-time implementation of the controller \eqref{ct-controller} is presented in~\cite{Koch2019,Koch2019a} by Koch \etal\ The authors utilize the matching approach to establish a discrete-time controller, resulting in the controller functions
	\begin{align}\label{dt-koch}
		\Psi_{1,\Koch} &=-\frac{1}{\alpha h} \left(e^{\frac{p_1 h}{\sqrt{|x_{1,k}|}}} + e^{\frac{p_2 h}{\sqrt{|x_{1,k}|}}} - 2\right) x_{1,k} - \frac{h \beta}{\alpha} \Psi_{2,\Koch}, \nonum
		\Psi_{2,\Koch} &=\frac{1}{h^2\beta} \left(e^{\frac{p_1 h}{\sqrt{|x_{1,k}|}}}-1\right) \left(e^{\frac{p_2 h}{\sqrt{|x_{1,k}|}}}-1\right) x_{1,k},
	\end{align}
	with $p_{1,2} = -\frac{\alpha}{2} \pm \sqrt{\frac{\alpha^2}{4} - \beta}$. The discrete-time closed-loop system is shown to avoid discretization-chattering effects and to be globally asymptotically stable in the disturbance-free case. Note that $\Psi_{1,\Koch}$ differs from the function in~\cite{Koch2019} due to the different definition of the general discrete-time controller \eqref{dt-general-controller}.
	
	\subsection{Semi-Implicit Discretization}
	The third known discrete-time version of the STC that is considered in this paper was published in~\cite{Xiong2022} by Xiong \etal\ and is obtained by a semi-implicit discretization. It consists of the controller functions
	\begin{align}\label{dt-xiong}
		\Psi_{1,\Xiong} &=\frac{1}{h \alpha} D_k \sat{\frac{x_{1,k}}{D_k}}, \nonum
		\Psi_{2,\Xiong} &=\sat{\frac{x_{1,k}}{D_k}},
	\end{align}
	where
	$
		D_k = \begin{cases} h \alpha |x_{1,k}|^{\frac{1}{2}} + h^2 \beta & \text{if } |x_{1,k}| > h \alpha |x_{1,k}|^{\frac{1}{2}} + h^2 \beta \\ h^2 \beta & \text{else}. \end{cases}
	$
	The authors show that the controller is insensitive to an overestimation of the gains regarding the asymptotic accuracy of the closed-loop system.
	
	\subsection{Low-Chattering Discretization}
	Finally, the last considered discrete-time STC is derived from the low-chattering differentiator presented in~\cite{Hanan2022}. The controller functions take the form
	
	\begin{align}\label{dt-hanan}
		\Psi_{1,\Hanan} &=\sat{\frac{|x_{1,k}|}{\revision{\gamma} h^2}}^\frac{1}{2} \sgnpow{x_{1,k}}{\frac{1}{2}} - \frac{h\beta}{\alpha} \sat{\frac{x_{1,k}}{\revision{\gamma} h^2}}, \nonum
		\Psi_{2,\Hanan} &=\sat{\frac{x_{1,k}}{\revision{\gamma} h^2}}.
	\end{align}
	The derivation of this discrete-time representation of the STC \revision{and the selection of $\gamma$} are given in Appendix~\ref{app:derivation-hanan}.

	Fig.~\ref{fig:controller-functions}a and~\ref{fig:controller-functions}b show the discrete-time controller functions $\Psi_{1,j}(x_1)$ respective $\Psi_{2,j}(x_1)$ from~\eqref{dt-brogliato},~\eqref{dt-koch},~\eqref{dt-xiong} and~\eqref{dt-hanan} with $j\in\{\Brogliato, \Koch, \Xiong, \Hanan\}$, as well as the functions $\sgnpow{x_1}{\frac{1}{2}}$ respective $\sgnpow{x_1}{0}$ from the continuous-time controller~\eqref{ct-controller}. The parameters were chosen as $h=1$, $\beta=1$, $\alpha=\sqrt{2\beta}$. For the computation of $\Psi_{1,\Brogliato}(x_1)$ and $\Psi_{2,\Brogliato}(x_1)$, $\nu_k$ was assumed to be zero. This is the case in steady state in the absence of a disturbance.
	Fig.~\ref{fig:controller-functions}a shows two regions within $x_1$, where $\Psi_{1,\Xiong}$ is constant. Between these regions, $\Psi_{1,\Xiong}$ is linear. The size of the constant regions depends on the parameters $\alpha$ and $\beta$ and is large \wrt\ the linear region, when $\alpha$ is large compared to $\beta$. The effect of this linear region will be discussed in Section~\ref{sec:simulation}.
	\revision{Note that Fig.~\ref{fig:controller-functions} helps to get an intuitive understanding of the controllers.}
	
	\begin{figure}
		\centering
		\includegraphics[width=\linewidth]{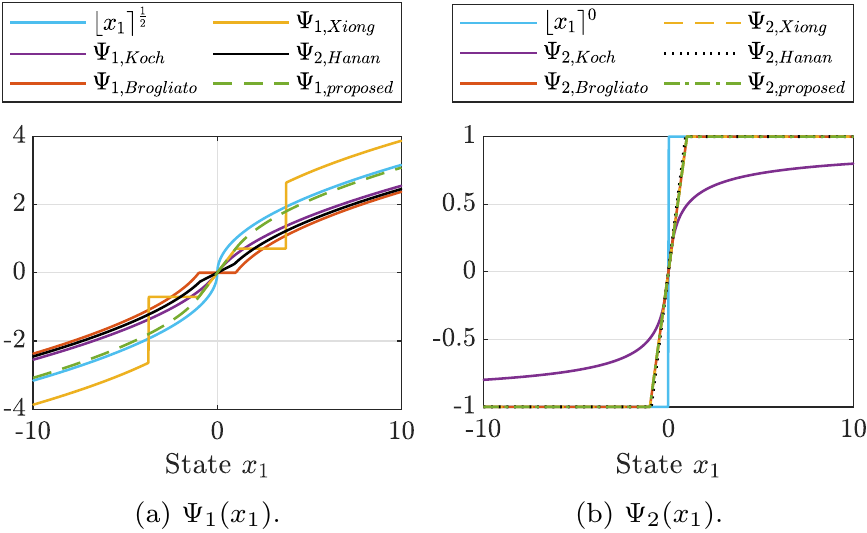}
		\caption{Compared controller functions.}
		\label{fig:controller-functions}
	\end{figure}
	
	\section{Proposed Discrete-Time STC}\label{sec:proposed}
	Sampling the continuous-time state $x_1(t)$ of system~\eqref{ct-plant} with $u(t) = u_k~\forall t\in[kh,(k+1)h)$ results in \revision{${x_1((k+1)h)} = x_1(kh) + h u_k + {\int_{kh}^{(k+1)h} \varphi(\tau)\mathrm{d}\tau}.$}
	\revision{Defining the discrete-time state $\varphi_k \coloneqq \frac{1}{h}\int_{kh}^{(k+1)h}\varphi(\tau)\mathrm{d}\tau$ and the unknown discrete-time input $\Delta_k \coloneqq \frac{1}{h}(\varphi_{k+1} - \varphi_k)$ yields
		\begin{align}~\label{dt-disturbance-integral}
			\Delta_k &= \frac{1}{h^2} \left(\int_{(k+1)h}^{(k+2)h}\varphi(\tau)\mathrm{d}\tau - \int_{kh}^{(k+1)h}\varphi(\tau)\mathrm{d}\tau\right) \nonum
			&= \frac{1}{h^2} \int_{kh}^{(k+1)h}\varphi(\tau+h)-\varphi(\tau)\mathrm{d}\tau.
		\end{align}
		As $|\dot \varphi(t)| = |\Delta(t)| \leq L~\forall t>0$, $\revision{|}\varphi(t+h)-\varphi(t)\revision{|} \leq Lh$ holds and thus $|\Delta_k| \leq L$. Defining $x_{1,k} \coloneqq x_1(kh)$ yields the discrete-time plant model
		\begin{align}\label{dt-plant}
			x_{1,k+1} &=x_{1,k} + h u_{k} + h \varphi_{k}, \nonum
			\varphi_{k+1} &=\varphi_k + h \Delta_{k}.
		\end{align}
		The state $\varphi_k$ can be interpreted as the mean value of $\varphi(t)$ in the interval $t \in [kh,(k+1)h)$\revision{. Also, } $\varphi_k$ as well as $\Delta_k$ are virtual values\revision{, and not samples of the continuous-time signals $\Delta(t)$ and $\varphi(t)$}.
		Note that \revision{even though}~\eqref{dt-plant} is structurally equivalent to an Euler forward discretization of~\eqref{ct-plant} the state $x_{1,k}$ coincides with the samples $x_1(kh)$.}
	
	%
	
	The novel discrete-time STC functions
	\begin{align}\label{dt-proposed}
		\Psi_{1,\proposed} &=\sign{x_{1,k}} \left( \frac{h \beta}{\alpha}\sat{\frac{|x_{1,k}|}{h^2 \beta}} - \frac{h\alpha}{2} + \right. \nonum
		&\quad\left. + \sqrt{\frac{h^2\alpha^2}{4} + \text{max}(0, |x_{1,k}| - h^2 \beta)}\right), \nonum
		\Psi_{2,\proposed} &=\sat{\frac{x_{1,k}}{h^2 \beta}}
	\end{align}
	are proposed in this paper for the discrete-time plant model~\eqref{dt-plant}. \revision{It is worth noting} that the controller functions~\eqref{dt-proposed} resemble the terms of the implicit controller functions~\eqref{dt-brogliato}, with $x_{1,k}+h\nu_k$ replaced by $x_{1,k}$, and extended by the saturation term in the first equation. The controller can therefore be regarded as a modified implicitly discretized STC.
	Fig.~\ref{fig:controller-functions}a and~\ref{fig:controller-functions}b also show the controller functions in~\eqref{dt-proposed}. Note that $\Psi_{2,\proposed}$, $\Psi_{2,\Brogliato}$ and $\Psi_{2,\Xiong}$ coincide. Furthermore, $\Psi_{1,\proposed}$ and $\Psi_{1,\Xiong}$ coincide near the origin, i.e. at $|x_1| \leq h^2\beta$.

	\begin{figure}[h]
		\centering
		\includegraphics[width=1\linewidth]{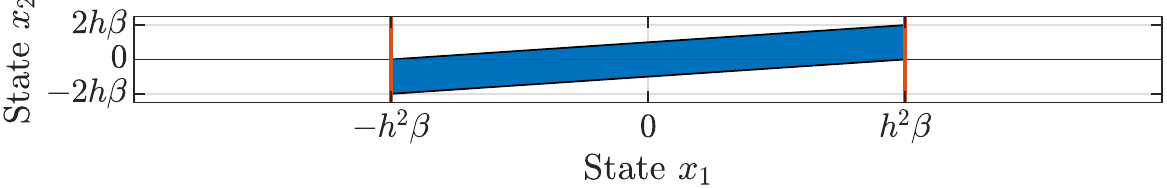}
		\caption{Invariant set $\mathcal{M}$ of the closed-loop system.}
		\label{fig:roc}
	\end{figure}
	
	Let the unknown virtual state $x_{2,k}$ be defined as $x_{2,k} \coloneqq \nu_k + \varphi_k$. The closed-loop system resulting from the discrete-time plant~\eqref{dt-plant} and the controller $\mathcal{C}_\proposed$ is
	\begin{align}\label{dt-closedloop-system}
		x_{1,k+1} &=x_{1,k} -2 h^2\beta \sat{\frac{x_{1,k}}{h^2\beta}} - h\alpha \sign{x_{1,k}}\left( -\frac{h \alpha}{2} + \right.\nonum
		&\quad\left.+\sqrt{\frac{h^2\alpha^2}{4} + \max(0, |x_{1,k}| - h^2\beta)} \right) + h x_{2,k}, \nonum
		x_{2,k+1} &=x_{2,k} - h \beta \sat{\frac{x_{1,k}}{h^2\beta}} + h \Delta_{k}.
	\end{align}
	
	In the following, the stability properties of the discrete-time STA~\eqref{dt-closedloop-system} are examined. \revision{For this define $\mathcal{M} = \{(x_{1,k},x_{2,k}) \in \mathbb{R}^2 | |x_{1,k}| \leq h^2 \beta, |hx_{2,k} - x_{1,k}| \leq h^2\beta\}$. $\mathcal{M}$ is plotted in Fig.~\ref{fig:roc} in state-space as a blue area.}
	
	\begin{proposition}\label{proposition-1}
		Consider the closed-loop system~\eqref{dt-closedloop-system} with the Lipschitz constant $L$, i.e. $|\Delta_k| \leq L~\forall k$\revision{, and $\beta > L$}. 
		\revision{Then $\mathcal{M}$ is a forward invariant set and} if \revision{$x_k \in \mathcal{M}$} is fulfilled for some $k = K$, the steady-state error \revision{is limited, i.e. $\limsup_{k\geq K+2} |x_{1,k}| \leq h^2 L$}. Further, the closed-loop system is exact in the absence of a disturbance, i.e. the state $x_{1,k}$ converges to zero. Thus, discretization-chattering is completely avoided.
	\end{proposition}
	
	\begin{proof}
		Assume $|x_{1,k}| \leq h^2 \beta$. Then \revision{the controller~\eqref{dt-general-controller} with}~\eqref{dt-proposed} can be simplified to
		\begin{align}\label{dt-proposed-vicinity}
			u_{k} &=\revision{-\frac{1}{h} x_{1,k} + \nu_{k+1},} \nonum
			\nu_{k+1} &=\nu_k - \frac{1}{h} x_{1,k}\revision{,}
		\end{align}
		\revision{which in explicit form yields $u_k = -\frac{2}{h} x_{1,k} + \nu_{k}$.}
		
		The second-order closed-loop system resulting from~\eqref{dt-plant} and~\eqref{dt-proposed-vicinity} \revision{in matrix-form} is then given by
		\begin{align}\label{dt-closedloop-vic}
			\begin{bmatrix} x_{1,k+1} \\ x_{2,k+1} \end{bmatrix} = \underbrace{\begin{bmatrix} -1 & h \\ -\frac{1}{h} & 1 \end{bmatrix}}_{\revision{M}} \begin{bmatrix} x_{1,k} \\ x_{2,k} \end{bmatrix} + \begin{bmatrix} 0 \\ h \end{bmatrix} \Delta_{k}.
		\end{align}
		The eigenvalues of the system-matrix \revision{$M$} are both zero, which means that~\eqref{dt-closedloop-vic} is a second-order dead-beat system~\cite{EmamiNaeini1982}. Thus, the steady state is reached after two steps.
		
		The discrete-time controller acts as a dead-beat controller, whenever $|x_{1,k}| \leq h^2 \beta$.
		In order to reach steady state, the dead-beat controller~\eqref{dt-proposed-vicinity} must be applied to the plant two times consecutively. Thus, to reach steady state $|x_{1,k}| \leq h^2 \beta$ and $|x_{1,k+1}| \leq h^2 \beta$ must hold. According to~\eqref{dt-closedloop-vic} \revision{$|x_{1,k+1}| \leq h^2 \beta \quad \Leftrightarrow \quad |h x_{2,k} - x_{1,k}| \leq h^2 \beta$}\revision{, which corresponds to $x_k \in \mathcal{M}$}.
		\revision{Assume $x_{K} \in \mathcal M$. Then~\eqref{dt-closedloop-vic} gives $|x_{1,K+2}| = {|-x_{1,K+1} + h x_{2,K+1}|} = |-(-x_{1,K}+hx_{2,K}) + (-x_{1,K}+hx_{2,K})+h^2\Delta_K| = |h^2\Delta_K| \leq h^2 L < h^2\beta$. Therefore, $\mathcal{M}$ is forward invariant for system~\eqref{dt-closedloop-system} and $x_{1,K+2+i} = h^2\Delta_{K+i}~\forall i\geq0$.}
	\end{proof}

	\begin{theorem}\label{theorem-1}
		\revision{Let $L \geq 0$, $|\Delta_k| \leq L~\forall k$, ${V>0}$ and discretization time $h > 0$. Given parameters 
		${\alpha > 0}$, ${\beta > \max\left(4L, 
		\frac{5}{7}\frac{\sqrt{V}}{h}, \sqrt{L^2+\frac{2L\sqrt{V}}{h^2}}\right)}$,
		and initial states $(x_{1,0}, x_{2,0})$ fulfilling $V \geq V_0 \coloneqq 2\beta |x_{1,0} - h x_{2,0}| + x_{2,0}^2$. Then, the states of system~\eqref{dt-closedloop-system} converge to the set $\mathcal{M}$ in finite time and remain there.
		}
		\revision{Further, in the absence of a disturbance, i.e. $L=0$, and if $\alpha>0$, $\beta>0$ the origin of system~\eqref{dt-closedloop-system} is globally finite-time stable.}
	\end{theorem}
	
	The proof of Theorem~\ref{theorem-1} is in Appendix~\ref{app:proof-theorem-1}.\qed

	\section{Evaluation in Simulation Studies}\label{sec:simulation}
	
	In this section, the results of numerical simulations are presented. 
	All simulations were performed in MATLAB\Rkreis/Simulink\Rkreis. 
	\revision{The plant was simulated in discrete-time according to~\eqref{dt-plant} with the same discretization-time $h$ as the controllers. The discrete-time disturbance $\Delta_k$ for the simulations was computed solving the integral in~\eqref{dt-disturbance-integral} analytically yielding $x_k = x(hk)$. The simulations were performed with a fixed-step solver.}
	In the following, $\Sigma_*$ denotes the simulated closed-loop system consisting of the plant and the discrete-time controller $\mathcal{C}_*$.
	It is shown in what regard $\mathcal{C}_\proposed$ is an improvement to state-of-the-art discrete-time STCs. 
	
	\subsection{\revision{Disturbed Case}}
	
	\begin{figure}
		\centering
		\includegraphics[width=\linewidth]{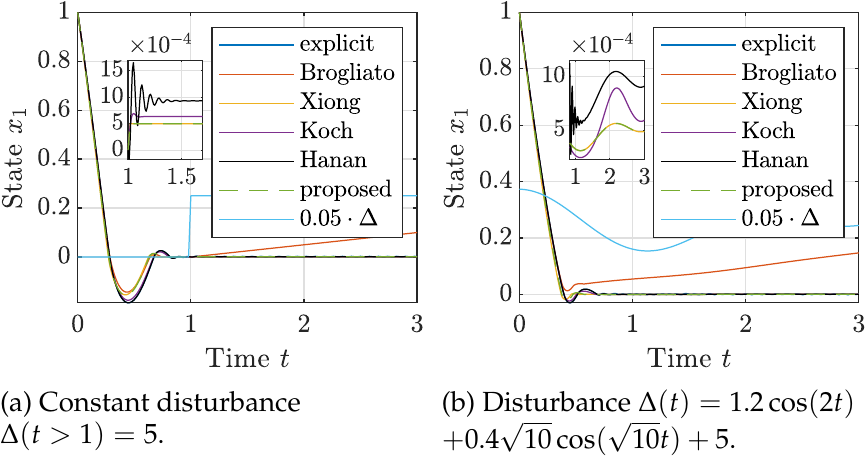}
		\caption{Simulation in time-domain with a disturbance, $\alpha = \sqrt{10}$, $\beta = 10$, $h = 0.01$ and $x_1(0) = 1$.}
		\label{sim:dist}
	\end{figure}
	The first simulation was performed with a constant disturbance $\Delta = 1~\forall t\geq1$ and $\Delta=0~\forall t<1$. The parameters where chosen as $\alpha = \sqrt{10}$, $\beta = 10$, \revision{the discretization time} $h = 0.01$ and $x_1(0) = 1$. The second simulation was performed with the same parameters and a disturbance known from literature~\cite{Koch2019a, Xiong2022} with a constant offset of $5$, $\Delta(t) = 1.2\cos(2t) + 0.4\sqrt{10}\cos(\sqrt{10}t) + 5$. The results are presented in Fig.~\ref{sim:dist}a and~\ref{sim:dist}b. The results clearly show that $\mathcal{C}_\Brogliato$ is not capable of rejecting constant parts of the disturbance $\Delta$, as it was described in Section~\ref{sec:related-work}. All other controllers result in a state $x_1$ converging close to zero. \revision{$\mathcal C_\Hanan$ leads to decaying oscillations in $x_1$.}
	
	\subsection{\revision{Undisturbed Case}}

	\begin{figure*}
		\centering
		\includegraphics[width=1\textwidth]{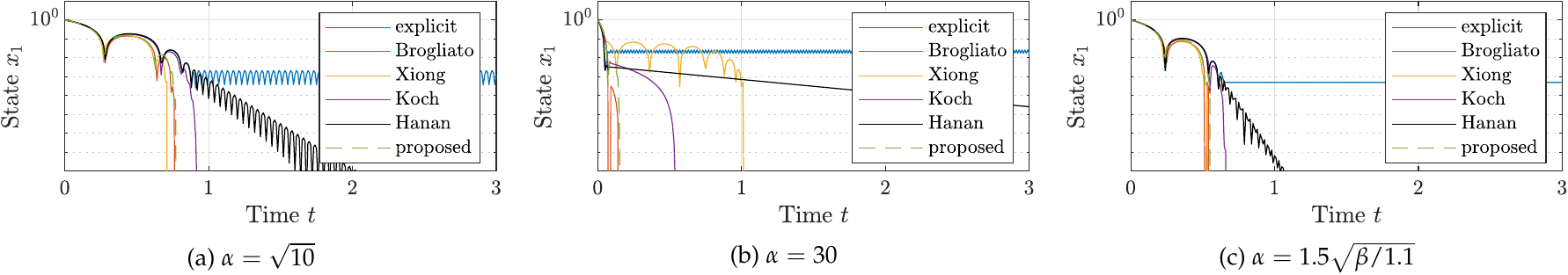}
		\caption{Simulations in time-domain in the undisturbed case, $\beta = 10$, $h = 0.01$ and $x_1(0) = 1$.}
		\label{sim:no-dist}
	\end{figure*}
	
	Three more simulations were performed with no disturbance, i.e. $\Delta \equiv 0$. In Fig.~\ref{sim:no-dist}a $\alpha = \sqrt{10}$ was chosen as before. In Fig.~\ref{sim:no-dist}b the parameter was set to $\alpha = 30$, which is large \wrt\ $\beta$, and in Fig.~\ref{sim:no-dist}c $\alpha = 1.5 \sqrt{\beta/1.1}$ was set according to the recommended parameter choice for $\mathcal{C}_\Hanan$ in~\cite[Fig.~3]{Hanan2022}. The other parameters remained unchanged, i.e. $\beta = 10$, \revision{discretization time} $h = 0.01$ and $x_1(0) = 1$. 
	The state $x_1$ is depicted in absolute values and scaled logarithmically in these plots, in order to emphasize the differences between the results of the controllers. The results show that $\Sigma_\explicit$ is not exact and exhibits discretization-chattering in steady state. \revision{All other} systems converge to zero without discretization-chattering effects. However, $\Sigma_\Hanan$ converges slower than the other systems.
	From the continuous-time STA~\eqref{ct-sta} it is expected that increasing the parameter $\alpha$ leads to faster convergence times. However, increasing $\alpha$ from $\sqrt{10}$ to $30$ in Fig.~\ref{sim:no-dist}b shows an increased convergence time of $\Sigma_\Xiong$. The systems $\Sigma_\proposed$ and $\Sigma_\Brogliato$ converge faster to zero. System $\Sigma_\Koch$ exhibits a larger convergence time than $\Sigma_\proposed$ and $\Sigma_\Brogliato$.
	
	\subsection{Convergence Time when Varying One Parameter}
	
	\begin{figure}
		\centering
		\includegraphics[width=\linewidth]{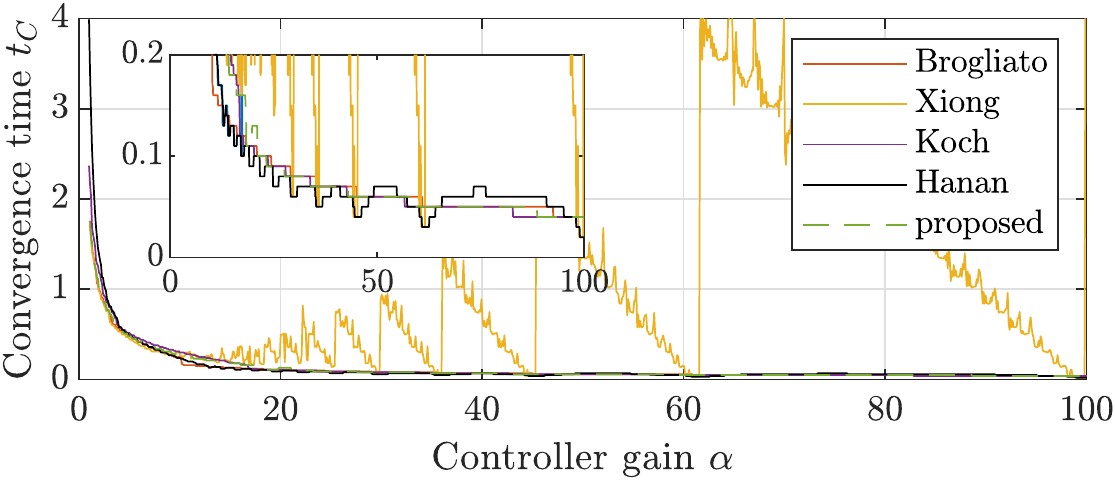}
		\caption{Convergence time $t_C$ over varying parameter $\alpha$.}
		\label{sim:tC-over-alpha}
	\end{figure}
	
	In order to analyze the behavior of increasing convergence times when increasing $\alpha$, the convergence time was determined for several parameter values $\alpha$. 
	\revision{The displayed convergence time $t_C$ is the lowest time for which the absolute state value does not exceed $1\%$ of the initial value, i.e. ${|x_1(t)| \leq 10^{-2} |x_1(0)|}~\forall t \geq t_C$.}
	Fig.~\ref{sim:tC-over-alpha} shows the convergence times of the compared systems over the parameter $\alpha$, which was set to values between $1$ and $100$, i.e. $0.1\beta$ and $10\beta$, in intervals of $0.1$. The other parameters were fixed at $\beta = 10$, $h = 0.01$ and $x_1(0) = 1$. 
	Fig.~\ref{sim:tC-over-alpha} illustrates that the convergence time of $\Sigma_\Xiong$ behaves very sensitive to changes in $\alpha$ when $\alpha > \beta$. Small changes in $\alpha$ can lead to a large increase of the convergence time, e.g. changing $\alpha$ from $29.8$ to $29.9$ results in $t_C$ changing from $0.11$ to $0.87$ (all numbers are rounded).
	The reason for the large convergence times of $\Sigma_\Xiong$ when $\alpha > \beta$ may be connected to the constant regions of the controller function $\Psi_{1,\Xiong}$ in~\eqref{dt-xiong} which is depicted in Fig.~\ref{fig:controller-functions}a. When $\alpha$ is large \wrt\ $\beta$, then the constant regions are larger \wrt\ the linear region of $\Psi_{1,\Xiong}$, as described in Section~\ref{sec:related-work}. The function $\Psi_{1,\Xiong}$ can be interpreted as a rate at which $x_1$ approaches the origin. When this function is constant with a rather small magnitude, this approaching phase could then take longer, the larger this constant region is.
	\revision{$\Sigma_\Hanan$ shows oscillations in the convergence time with very small magnitude, interestingly with a similar frequency as the oscillations in $\Sigma_\Xiong$.}
	
	\subsection{Steady-State Accuracy when Varying the Parameters}
	
	\begin{figure*}
		\centering
		\includegraphics[width=1\textwidth]{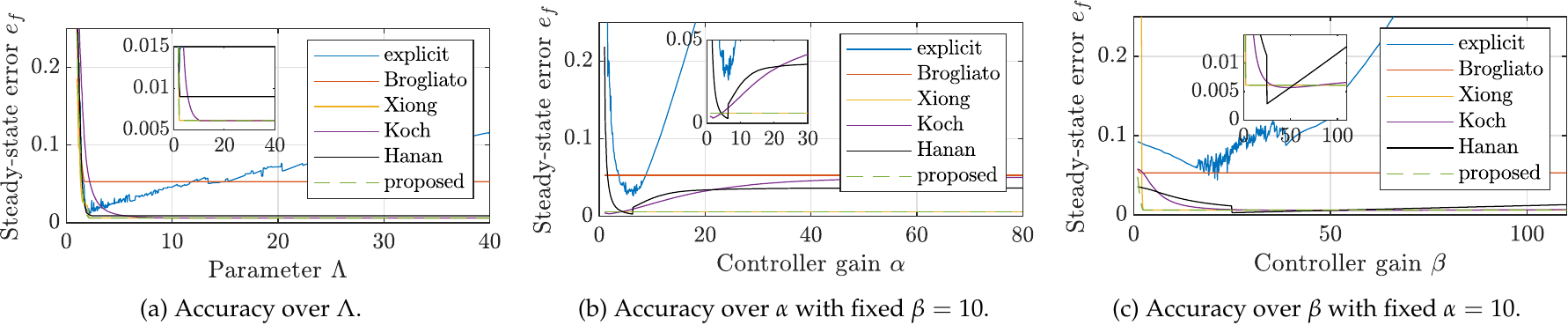}
		\caption{Accuracy by means of the steady-state error $e_f$, $h=0.05$ and $x_1(0) = 0$.}
		\label{sim:ss-error}
	\end{figure*}
	
	Finally, simulations were performed regarding the accuracy of the closed-loop systems, i.e. the remaining steady-state error. Let this steady-state error be defined as $e_f = \limsup_{t} {|x_1|}$, with the initial value $x_1(0) = 0$. The systems were simulated until $t=20$. The disturbance was chosen as $\Delta(t) = 1.2\cos(2t) + 0.4\sqrt{10}\cos(\sqrt{10}t)$, which was also used in~\cite{Koch2019a, Xiong2022}. The discretization time was set to $h=0.05$. Fig.~\ref{sim:ss-error}a depicts $e_f$ over the value $\Lambda$, which determines the parameters $\alpha=1.5 \sqrt{\Lambda}$ and $\beta = 1.1 \Lambda$. This parameter relation corresponds to the suggested parameter choice in~\cite{Hanan2022} and was also applied in~\cite{Xiong2022}, where the same simulation was performed. The value $\Lambda$ was varied between $1$ and $40$ in $1000$ steps. The results from~\cite{Xiong2022} were reproduced, whereas $\Sigma_\proposed$ and $\Sigma_\Hanan$ performed very similar to $\Sigma_\Xiong$. Fig.~\ref{sim:ss-error}b and~\ref{sim:ss-error}c show $e_f$ over varying $\alpha$ and $\beta$ from $1$ to $80$ and $110$, respectively, in $1000$ steps. The second controller parameter was fixed at $10$. The results again show the exactness of $\Sigma_\proposed$ as well as $\Sigma_\Xiong$, which yield the same small steady-state error after some minimal gain values. The system  $\Sigma_\Brogliato$, however, settles at larger errors, due to the appearance of the attenuated perturbation $\varphi$ in $x_1$ in steady state. In the results of $\Sigma_\Koch$, which is not the result of an implicit approach, a dependency between the controller parameters $\alpha$ and $\beta$ and the steady-state error can be observed. Increasing $\alpha$ even drives the steady-state error of $\Sigma_\Koch$ to the same level where $\Sigma_\Brogliato$ settles, as can be seen in Fig.~\ref{sim:ss-error}b.
	The system $\Sigma_\Hanan$ achieves \revision{smaller} steady-state errors \revision{than $\Sigma_\proposed$ and $\Sigma_\Xiong$} within some bounds of $\alpha$ and $\beta$. Outside of these bounds, the accuracy of $\Sigma_\Hanan$ \revision{decreases significantly.}
	
	\subsection{State Trajectories}
	
	\begin{figure}
		\centering
		\includegraphics[width=0.8\linewidth]{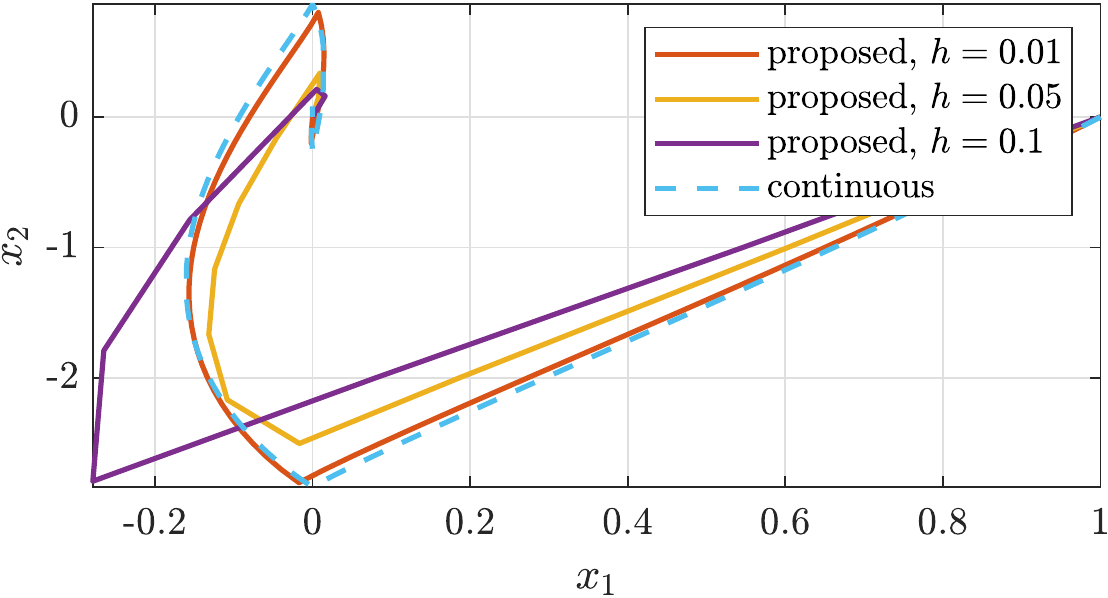}
		\caption{State trajectories with $\Delta \equiv 0$.}
		\label{sim:trajectories}
	\end{figure}

	\revision{	Figure~\ref{sim:trajectories} shows the state trajectories of the proposed algorithm with various discretization times $h$ comparative to the continuous-time algorithm.
	In this simulation, the input $\Delta \equiv 0$ and the parameters were selected as ${\alpha = \sqrt{10}}$, $\beta = 10$, $x_{1,0} = 1$, $x_{2,0} = 0$ and $h\in\{0.01, 0.05, 0.1\}$.
	The continuous-time trajectory was established with an Euler forward discretized STC and a discretization time of $10^{-5}$. 
	In this simulation $\Sigma_\proposed$ yields results similar to the continuous-time trajectory for small $h$. }
	
	\section{Conclusion}\label{sec:conclusion}
	In this paper, a novel discrete-time super-twisting controller is presented.
	It is shown to \revision{converge to an invariant set in finite time}. Additionally, the absence of any discretization-chattering effects is shown.
	The controller is directly compared to previously published discrete-time super-twisting controllers analytically regarding the controller structure as well as in simulation studies. The \revision{analytic considerations and} simulations showed that the presented controller resembles best \revision{several} properties of the continuous-time controller. 
	The proposed discretization \revision{can handle all Lipschitz-continuous perturbations.} Further, its finite convergence time decreases when any of the controller gains is increased. Moreover, the presented controller yields a steady-state error that is independent of the controller gains and it introduces no discretization-chattering effects.
	The proposed discrete-time super-twisting controller unites all of these beneficial properties, in contrast to the known controllers. In future work, the presented controller will be applied to real-world problems.
	
	\appendices
	
	\section{Derivation of the Low-Chattering Disc.}\label{app:derivation-hanan}
	The discrete-time differentiator according to~\cite{Hanan2022} with a differentiation order of $1$ and a filtering order of $0$ is given by
	\begin{align}~\label{differentiator}
		z_{0,k+1} &= z_{0,k} + h z_{1,k} - h \tilde\lambda_1 \hat L^{\frac{1}{2}} \sgnpow{z_{0,k} - f_{0,k}}{\frac{1}{2}} \nonum
		z_{1,k+1} &= z_{1,k} - h \tilde\lambda_0 \hat L \sgnpow{z_{0,k} - f_{0,k}}{0},
	\end{align}
	where $f_{0,k}$ is the discrete-time signal to be differentiated, \revision{$\tilde\lambda_1$} and \revision{$\tilde\lambda_0$} are constant parameters and $z_{0,k}$ and $z_{1,k}$ are the observer states that estimate the signal $f_{0,k}$ and its first derivative $f_{1,k}$, respectively. Further, $\hat L$ is an adaptive parameter following the computation law \revision{$\hat L = L \cdot \sat{\frac{|z_{0,k} - f_{0,k}|}{L k_L h^2}},$}
	with $L$ being the known Lipschitz constant of the unknown signal $f_{1,k}$ and a constant parameter $k_L$. Selecting $\tilde \lambda_0 L = \beta$, $\tilde \lambda_1 L^\frac{1}{2} = \alpha$ and $L k_L = \revision{\gamma}$ yields the error dynamics of the differentiator~\eqref{differentiator}
	\begin{align}\label{error-dynamics-hanan}
		x_{1,k+1} &= x_{1,k} + h x_{2,k} - h \alpha \sat{\frac{|x_{1,k}|}{\revision{\gamma} h^2}}^\frac{1}{2} \sgnpow{x_{1,k}}{\frac{1}{2}} \nonum
		x_{2,k+1} &= x_{2,k} - h \beta \sat{\frac{|x_{1,k}|}{\revision{\gamma} h^2}} \sgnpow{x_{1,k}}{0},
	\end{align}
	with $x_{1,k} = z_{0,k} - f_{0,k}$ and $x_{2,k} = z_{1,k} - f_{1,k}$.
	
	Therefore, by setting \revision{$u_k =- \alpha \sat{\frac{|x_{1,k}|}{\revision{\gamma} h^2}}^\frac{1}{2} \sgnpow{x_{1,k}}{\frac{1}{2}} + h\beta \sat{\frac{x_{1,k}}{\revision{\gamma} h^2}} + \nu_{k+1}$, $\nu_{k+1} =\nu_k - h \beta \sat{\frac{x_{1,k}}{\revision{\gamma} h^2}}$,}
	the closed-loop dynamics of system~\eqref{dt-plant} will follow the dynamics~\eqref{error-dynamics-hanan}, which yields the controller functions in~\eqref{dt-hanan}.
	\revision{This controller has a third tuning parameter $\gamma$. System~\eqref{error-dynamics-hanan} is linear in the band of the saturation around the origin. So, a natural choice of $\gamma$ is such that the eigenvalues of this linear system are in the unit disk. This can be achieved by e.g. selecting ${\gamma = G\begin{cases} \beta^2/\alpha^2 \text{ if } \alpha < 2 \sqrt{\beta}\\ \alpha^2/4 \text{ else,}\end{cases}}$ with $G>1$ which is used in Sections~\ref{sec:related-work} and~\ref{sec:simulation} with $G=1.5^2/1.1^2\approx1.8595$. This value was chosen, as for the relation $\alpha=1.5\sqrt{\Lambda}$, $\beta = 1.1\Lambda$ this yields $\gamma=L$, respective $k_L=1$, which is the recommended parameter value in~\cite{Hanan2022}.}
	
	\section{Proof of Theorem~\ref{theorem-1}}\label{app:proof-theorem-1}
		\revision{If $x_k \in \mathcal{M}$ for some $k$, then Proposition~\ref{proposition-1} applies, and $x_k$ remains in $\mathcal{M}$.
		Otherwise it is proven that $x_k$ converges to $\mathcal{M}$ in finite time. 
		In the following, it is assumed that $x_k \notin \mathcal{M}$.}
		Inspired by the stability analysis in~\cite[Theorem~IV.1]{Koch2019} define the Lyapunov candidate $V_k = 2\beta |x_{1,k} - h x_{2,k}| + x_{2,k}^2$. In general, using~\eqref{dt-general-controller} 
		the next step of the Lyapunov candidate computes to
	$
			{V_{k+1} = 2\beta |x_{1,k} - h\alpha \Psi_1(x_{1,k})| + (x_{2,k} - h \beta \Psi_2(x_{1,k}) \revision{+ h\Delta_k})^2.}$
		For the proposed controller, this yields
		\begin{align}
			V_{k+1} &=2\beta \left|x_{1,k} - h^2\beta \sat{\frac{x_{1,k}}{h^2\beta}} - h\alpha \sign{x_{1,k}} \cdot \right.\nonum
			&\quad\cdot\left.\left( -\frac{h \alpha}{2} + \sqrt{\frac{h^2\alpha^2}{4} + \max(0, |x_{1,k}| - h^2\beta)} \right)\right| + \nonum
			&\quad+ \left(x_{2,k} - h \beta \sat{\frac{x_{1,k}}{h^2\beta}} \revision{+ h\Delta_k}\right)^2.
		\end{align}
		
		\begin{textcases}
			\begin{case}{$|x_{1,k}| \leq h^2\beta$}{} gives
				\begin{align}
					V_{k+1} &= 2\beta \left|x_{1,k} - x_{1,k} - h\alpha \sign{x_{1,k}} \left( -\frac{h \alpha}{2} + \sqrt{\frac{h^2\alpha^2}{4}} \right)\right| \nonum
					&\quad + \left(x_{2,k} - \frac{1}{h} x_{1,k} \revision{+ h\Delta_k}\right)^2 = \left(x_{2,k} - \frac{1}{h} x_{1,k} \revision{+ h\Delta_k}\right)^2.
				\end{align}
				The difference $\Delta V_k = V_{k+1} - V_k$ then computes to
				\begin{align}
					\Delta V_k &= -2\beta |x_{1,k} - h x_{2,k}| + \left(x_{2,k} - \frac{1}{h} x_{1,k}\revision{+ h\Delta_k}\right)^2 - x_{2,k}^2 = \nonum
					&= -2\beta |x_{1,k} - h x_{2,k}| + \frac{x_{1,k}}{h^2} (x_{1,k} - h x_{2,k}) - \frac{x_{1,k}x_{2,k}}{h} + \nonum
					&\quad \revision{ + h^2\Delta_k^2 - 2\Delta_k (x_{1,k} - h x_{2,k})}\nonum
					&\revision{\leq (-\beta + 2L) |x_{1,k} - h x_{2,k}| - \frac{x_{1,k}x_{2,k}}{h} + h^2L^2},
				\end{align}
				\revision{with the limit established with $|x_{1,k}| \leq h^2\beta$ and $|\Delta_k| \leq L$. 
					From $x_k \notin \mathcal{M}$ and $|x_{1,k}| \leq h^2\beta$ follows $|x_{1,k+1}| = |x_{1,k} - hx_{2,k}| > h^2\beta$ and $|x_{2,k}| > 0$.
					Two cases are distinguished.}
				
				\begin{subcase}{$\sign{x_{1,k}} = \sign{x_{2,k}}$ or $x_{1,k} = 0$}{case1a} yields
					\begin{align}
						\Delta V_k &\leq \left( -\beta \revision{ + 2L}\right) \underbrace{|x_{1,k} - h x_{2,k}|}_{> h^2\beta} - \frac{|x_{1,k}x_{2,k}|}{h} \revision{+ h^2L^2} \leq\nonum
						&\revision{\leq (-\beta + 2L)h^2\beta + h^2L^2 = h^2(L^2 - \beta^2 + 2L\beta) < 0},
					\end{align}
					where the last inequality is fulfilled due to \revision{$\beta > 4L$}.
				\end{subcase}
			
				\begin{subcase}{$\sign{x_{1,k}} = -\sign{x_{2,k}}$ and $x_{1,k} \neq 0$}{case_1b} \\
					This case gives $|x_{1,k} - hx_{2,k}| > |hx_{2,k}|$. Assume $L=0$. Then, with $|x_{1,k}| \leq h^2\beta$
					\begin{align}
						\Delta V_k &< -h\beta |x_{2,k}| + \frac{|x_{1,k}x_{2,k}|}{h}\leq -h\beta |x_{2,k}| + h\beta|x_{2,k}| = 0.
					\end{align}
					Now, assume ${L>0}$, which gives with ${H \coloneqq  -\frac{h^2\alpha^2}{2} + \sqrt{\frac{h^4\alpha^4}{4} + h^2\alpha^2 (|x_{1,k+1}| - h^2\beta)}  > 0}$
					\begin{align}
						V_{k+2} 
						&= 2\beta\left||x_{1,k+1}| - h^2\beta - H\right| + (x_{2,k+1} - \sgnpow{x_{1,k+1}}{0}h\beta)^2 +\nonum 
						&\quad + 2\Delta_{k+1}(hx_{2,k+1} - \sgnpow{x_{1,k+1}}{0}h^2\beta) + h^2\Delta_{k+1}^2.
					\end{align}
					With $x_{1,k+1} = -x_{1,k} + hx_{2,k}$ and $x_{2,k+1} = -\frac{x_{1,k}}{h} + x_{2,k} + h\Delta_k$ from~(14) and $\sgnpow{x_{1,k+1}}{0} = - \sgnpow{x_{1,k}}{0}$ from $\sign{x_{1,k}} = -\sign{x_{2,k}}$ this further computes to
					\begin{align}
						V_{k+2} &= 2\beta\left|-|x_{1,k}| - h|x_{2,k}| + h^2\beta + H\right| +\nonum
						&\quad + \left(-\frac{x_{1,k}}{h} +x_{2,k} + \sgnpow{x_{1,k}}{0}h\beta + h\Delta_k\right)^2 +\nonum
						&\quad + 2\Delta_{k+1}(-x_{1,k}+hx_{2,k}+\sgnpow{x_{1,k}}{0}h^2\beta + h^2\Delta_k) + h^2\Delta_{k+1}^2
					\end{align}
					Without loss of generality assume $x_{1,k} > 0$, i.e. $x_{2,k} < 0$, in the remaining part of this case. Thus, with $|\Delta_k| \leq L$, $|\Delta_{k+1}| \leq L$ and $-x_{1,k} + hx_{2,k} + h^2\beta < 0$ the upper limit
					\begin{align}
						V_{k+2} &\leq 2\beta\left|hx_{2,k} -x_{1,k} + h^2\beta + H\right| + \left(x_{2,k} -\frac{x_{1,k}}{h} + h\beta - hL\right)^2 - \nonum
						&- 2L (hx_{2,k}-x_{1,k}+h^2\beta - h^2L) + h^2L^2
					\end{align}
					is established.
					It can easily be shown that $-x_{1,k} + hx_{2,k} + h^2\beta + H < 0$.
%
					With some rearranging steps this gives
					\begin{align}
						&V_{k+2} - V_k \leq 
						 -2\beta (h^2\beta + H) + \frac{x_{1,k}^2}{h^2} + 2\frac{|x_{1,k}x_{2,k}|}{h}+\nonum
						&\qquad + (\beta-L)\left(-2x_{1,k} + 2hx_{2,k} + h^2\beta - h^2L\right) + \nonum
						&\qquad + L 2 x_{1,k} - L 2hx_{2,k} - L 2 h^2\beta + L 3h^2L. 
					\end{align}
					With $|x_{1,k}| \leq h^2\beta$ and using the last expression the difference is further limited by
					\begin{align}
						&V_{k+2} - V_k \leq -2\beta (h^2\beta + H) + \beta x_{1,k} - 2h\beta x_{2,k} + \nonum
						&\qquad + 2(\beta-2L)(hx_{2,k}-x_{1,k}) + h^2\beta(\beta-3L) - h^2L(\beta-4L).
					\end{align}
					Using $2(\beta-2L) = (\beta-3L) + \beta -L$ the last expression can be rewritten as
					\begin{align}
						&V_{k+2} - V_k \leq 
						-2\beta (h^2\beta + H) + L x_{1,k} - h(\beta+L) x_{2,k} + \nonum
						&\qquad + (\beta - 3L) (-x_{1,k} + hx_{2,k} + h^2\beta) - h^2L (\beta-4L).
					\end{align}
					Due to \revision{$\beta > 4L$} the term $h^2L(\beta-4L) < 0$, and due to $x_k \notin \mathcal M$ the term $(\beta-3L)(-x_{1,k} + hx_{2,k} + h^2\beta) < 0$. It must be shown that the sum of the remaining terms is negative, i.e. $2\beta H \geq -2h^2\beta^2 + L x_{1,k} - h(L+\beta) x_{2,k}$.
					This is trivial if the right-hand side is negative or zero, i.e. $|x_{2,k}| \leq \frac{h\beta (2\beta - L)}{\beta+L}$.
					It holds that $|x_{2,k}| \leq \sqrt{V_k} \leq \sqrt{V_0} \leq \sqrt{V}$. With \revision{$\beta > 4L$} and \revision{$\beta > \frac{5}{7} \frac{\sqrt{V}}{h}$} we have $|x_{2,k}| \leq \sqrt{V} < \frac{7}{5}h\beta \leq \frac{h\beta (2\beta -L)}{\beta + L}$ which was to be shown.
						Therefore $V_{k+2} < V_k$.
				\end{subcase}
			\end{case}
			
			\begin{case}{$|x_{1,k}| > h^2\beta$}{} yields
				\begin{align}
					&V_{k+1} = 2\beta \left|x_{1,k} - \sign{x_{1,k}} \left(h^2\beta + h\alpha \left( -\frac{h \alpha}{2} + \right.\right.\right. \nonum
					&\quad\left.\left.\left. + \sqrt{\frac{h^2\alpha^2}{4} + |x_{1,k}| - h^2\beta} \right) \right)\right| + \left(x_{2,k} - \sign{x_{1,k}} h\beta\right)^2 - \nonum
					&\quad \revision{ - 2\Delta_k(x_{2,k} - \sign{x_{1,k}} h\beta) + h^2\Delta_k^2}.
				\end{align}\\
				Let us introduce $z_{1,k} \coloneqq x_{1,k} - \sign{x_{1,k}} h^2\beta$, $z_{2,k} \coloneqq x_{2,k} - \sign{x_{1,k}} h\beta$, with $z_{1,k} \in \mathbb{R}\backslash\{0\}$ and $z_{2,k} \in \mathbb{R}$. Note that $\sign{x_{1,k}} = \sign{z_{1,k}}$. This gives
				\begin{align}
					V_k &=2\beta |z_{1,k} - h z_{2,k}| + z_{2,k}^2 + \sign{z_{1,k}} 2 h\beta z_{2,k} + h^2\beta^2, \nonum
					V_{k+1} &=2\beta \left|z_{1,k} - h\alpha \sign{z_{1,k}} \left(-\frac{h \alpha}{2} + \sqrt{\frac{h^2\alpha^2}{4} + |z_{1,k}|} \right)\right| \nonum
					&\quad + z_{2,k}^2 \revision{ - 2\Delta_kz_{2,k} + h^2\Delta_k^2},
				\end{align}
				and further \revision{with $A\coloneqq h\alpha \left(-\frac{h \alpha}{2} + \sqrt{\frac{h^2\alpha^2}{4} + |z_{1,k}|} \right) > 0$}
				\begin{align}
					\Delta V_k &= 2\beta \left( \left|z_{1,k} - \sign{z_{1,k}} \revision{A}\right| - \left| z_{1,k} - h z_{2,k} \right| \right)\nonum
					&\quad - 2 h \beta \sign{z_{1,k}} z_{2,k} - h^2\beta^2 \revision{ + h^2\Delta_k^2 - 2\Delta_k z_{2,k}} \nonum
					&\revision{\leq 2\beta \left( \left|z_{1,k} - \sign{z_{1,k}} \revision{A}\right| - \left| z_{1,k} - h z_{2,k} \right| \right)}\nonum
					&\quad \revision{+ 2L |z_{2,k}|- 2 h \beta \sign{z_{1,k}} z_{2,k} - h^2\beta^2 + h^2L^2}
				\end{align}
				Note that \revision{this upper limit of } $\Delta V_k$ is an even function.
				Therefore, it is sufficient to only consider the case ${z_{1,k} > 0}$.
				\revision{The following shows that $z_{1,k} - A \geq 0$ always holds, as}
				\begin{align}
					&\revision{\left(z_{1,k} + \frac{h^2\alpha^2}{2}\right)^2 \geq \left( \frac{h^4\alpha^4}{4} + h^2\alpha^2 z_{1,k} \right)} \quad
					\Leftrightarrow \quad\revision{z_{1,k}^2 \geq 0,}
				\end{align}
				\revision{which holds $\forall z_{1,k}$.} \revision{Together this yields}
				\begin{align}\label{DeltaVk_case2_undisturbed}
					\Delta V_k &\revision{\leq} 2\beta ( (z_{1,k} - \revision{A}) - \left| z_{1,k} - h z_{2,k} \right| )-2 h \beta z_{2,k} \nonum
					&\quad \revision{ + 2L|z_{2,k}|} - h^2\beta^2 \revision{+ h^2L^2},\qquad
				\end{align}
				and is assumed in the remaining part of the proof.
				\revision{Two cases are distinguishing in the following.}
				
				\begin{subcase}{$z_{1,k} - h z_{2,k} \geq 0$}{case2a} leads to
					\begin{align}
						&\Delta V_k \revision{\leq}2\beta (h z_{2,k} - A - hz_{2,k}) \revision{ + 2L|z_{2,k}|} - h^2 (\beta^2 - L^2) \nonum
						&\quad= - 2 \beta A \revision{ + 2L |z_{2,k}|} - h^2 (\beta^2 - L^2) 
						\leq \revision{2L |z_{2,k}| - h^2 (\beta^2 - L^2),}
					\end{align}
					\revision{which must be negative. For $L=0$ this is fulfilled. For $L>0$ and with $|z_{2,k}| \leq \sqrt{V_k} < \sqrt{V}$ it is sufficient that $2L\sqrt{V} < h^2(\beta^2-L^2)$,}
					\revision{which is fulfilled due to $\beta^2 > L^2+\frac{2L\sqrt{V}}{h^2}$ and thus $\Delta V_k < 0$.}
				\end{subcase}
				
				\begin{subcase}{$z_{1,k} - h z_{2,k} < 0$}{} gives $z_{2,k} > 0$ and
					\begin{align}
						\Delta V_k &\leq2\beta (2 z_{1,k} - A - h z_{2,k})  \revision{ + 2(L-h\beta)z_{2,k}} - h^2\beta^2 \revision{+ h^2L^2} \nonum
						&= 4\beta(z_{1,k}-hz_{2,k})-2\beta A+2Lz_{2,k} - h^2(\beta^2-L^2)\nonum
						&\leq 2Lz_{2,k} - h^2(\beta^2-L^2),
 					\end{align}
 					\revision{which was already shown in Case~\ref{case2a} to be negative.}
				\end{subcase}

			\end{case}
		\end{textcases}

		In the cases above all possible combinations of states $(x_{1,k}, x_{2,k})$ were considered.
		In all cases \revision{but Case~\ref{case_1b}} $(x_{1,k}, x_{2,k}) \notin \mathcal{M}$, $\Delta V_k(x_{1,k}, x_{2,k}) < 0$ was proven. \revision{In Case~\ref{case_1b} $V_{k+2}-V_k < 0$ was proven.
		Therefore, $\Delta V_{2,k} \coloneqq V_{k+2} - V_k = \Delta V_{k+1} + \Delta V_k < 0$. $V_k$ is continuous in ($x_{1,k}, x_{2,k}$) and as $\Delta V_k$ is continuous, also $\Delta V_{2,k}$ is continuous. Due to the continuity 
		the maximum of 
		$\Delta V_{2,k}$ exists. Further $V_k > 0$ $\forall x_k \neq 0$ and $V_k = 0$ for $x_k = 0$. So, $\exists V_M > 0$ such that $x_k \notin \mathcal{M} \Rightarrow V_k > V_M$.  Thus, $\exists \delta \coloneqq \max_{x_k \notin \mathcal{M}, V_k \leq V}(\Delta V_{2,k}) < 0$ and so the maximum number of steps until $\mathcal{M}$ is reached, $(V_0-V_M) / |\delta| < \infty$, is finite. Therefore, $x_k$ converges to $\mathcal{M}$ in a finite number of steps, i.e. in finite time.
	}

	\normalem
	\printbibliography
	
\end{document}